\documentclass[10pt, conference, a4paper]{IEEEtran}\IEEEoverridecommandlockouts
\usepackage{amsthm}
\usepackage{amssymb}
\usepackage{amsmath}
\usepackage{amsfonts}
\usepackage{graphicx}
\usepackage{algorithm}
\usepackage{algorithmic}
\usepackage{epstopdf}
\usepackage{cite}
\usepackage{hyperref}
\usepackage{subfigure}

\newtheorem{lemma}{\textbf{Lemma}}

\newtheorem{proposition}{\textbf{Proposition}}

\begin{document}

\title{Optimal Task Scheduling in Communication-Constrained Mobile Edge Computing Systems for Wireless Virtual Reality}
\author{Xiao Yang$^*$, Zhiyong Chen$^*$, Kuikui Li$^*$, Yaping Sun$^*$ and Hongming Zheng$^\dagger$\\
$^*$Cooperative Medianet Innovation Center,  Shanghai Jiao Tong University,  Shanghai,  P. R. China\\
$\dagger$Intel Corporation, P.R. China\\
Email: {\{yangxiao1652,  zhiyongchen, kuikuili, yapingsun\}@sjtu.edu.cn}, hongming.zheng@intel.com}
\maketitle

\begin{abstract}

Mobile edge computing (MEC) is expected to be an effective solution to deliver 360-degree virtual reality (VR) videos over wireless networks. In contrast to previous computation-constrained MEC framework, which reduces the computation-resource consumption at the mobile VR device by increasing the communication-resource consumption, we develop a communications-constrained MEC framework to reduce communication-resource consumption by increasing the computation-resource consumption and exploiting the caching resources at the mobile VR device in this paper. Specifically, according to the task modularization, the MEC server can only deliver the components which have not been stored in the VR device, and then the VR device uses the received components and the corresponding cached components to construct the task, resulting in low communication-resource consumption but high delay. The MEC server can also compute the task by itself to reduce the delay, however, it consumes more communication-resource due to the delivery of entire task. Therefore, we then propose a task scheduling strategy to decide which computation model should the MEC server operates, in order to minimize the communication-resource consumption under the delay constraint. Finally, we discuss the tradeoffs between communications, computing, and caching in the proposed system.
\end{abstract}

\section{Introduction}
Virtual reality (VR) over wireless networks is gaining an unprecedented attention due to the ability of bringing immersive experience to users. VR application is computational-intensive, communications-intensive and delay-sensitive, leading to the fact that most of
VR are wired with cables. Current wireless systems (e.g., LTE) cannot cope with the ultra-low latency and ultra-high throughput requirements of wireless VR application (e.g., 360-degree VR video) \cite{VR_Magazine}.

To address this problem, an effective solution is mobile edge computing (MEC)\cite{3C, MEC_BUPT}, which enables cloud computing capabilities at the edge of wireless network, e.g., base stations (BS). By deploying computation resource at the network edge, MEC performs the computation tasks closer to the VR device, improving the quality of computation experience.
\begin{figure}
\centering
\includegraphics[width=3.3in]{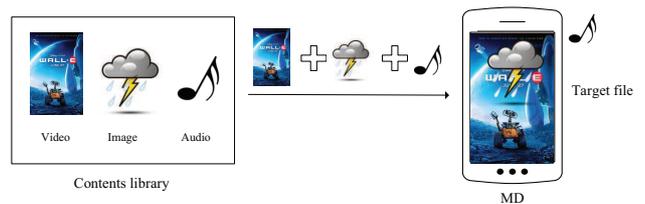}
\caption{Task can be modularized and then recovered at the user \cite{MMT,VR360}.}\label{task}
\vspace{-4.5mm}
\end{figure}

\textbf{Computation-Constrained MEC}, where the VR user can offload the computation tasks from the VR device to the MEC server due to the limited computation capability at the VR device, has attracted significant attention recently\cite{MECpowerdelay2, MEC_TCOM, huangkaibin2}. In \cite{MECpowerdelay2},  a dynamic computation scheduling algorithm based on Lyapunov theory was proposed to minimize the execution delay and task failure. Similarly, \cite{MEC_TCOM} proposed an optimization offloading framework to minimize the execution latency and the user's energy consumption. In order to minimize the mobile energy consumption under the constraint of execution delay, \cite{huangkaibin2} optimized the offloading policy based on the channel gains and the consumption of local computing energy. In general, the computation-constrained MEC exploits the communication resource to reduce the requirement of the user's computation resource. Therefore, it is quite suitable for the computational-intensive and delay-sensitive application with low bandwidth consumption, e.g., online picture processing.

The video data rate on online 360-degree VR videos is already many times greater than that of a high definition TV, even though 5G networks maybe not able to satisfy this requirement \cite{VR_Magazine}. Meanwhile, the end-to-end latency through the system is no more than $20~ms$ for online 360-degree VR videos. We can also use the MEC solution to improve network responsiveness and reduce latency, however, we cannot try to reduce the computation resource by increasing the communications resource overhead. In contrast, we should try to reduce the consumption of communications resource by taking advantage of the computation and caching resources at the mobile VR devices, where we term this solution as \textbf{Communications-Constrained MEC} in this paper.

In this paper, we present a communications-constrained MEC framework and develop an optimal task scheduling policy to minimize the average transmission data per task in this system, where the average transmission data per task is adopted as the consumption of communications resource metric. Our major contributions are summarized as follows:
\begin{itemize}
  \item \emph{We propose a communications-constrained MEC framework to reduce the consumption of communications resource by exploiting the caching resource and increasing the consumption of computing resource at the mobile VR device}. To address this issue, we consider that a 360-degree VR video (task) can be modularized today \cite{MMT,VR360}, e.g., MPEG Media Transport (MMT) standard, as shown in Fig. \ref{task}. Besides, the contents composing each whole task have popularity. As a result, for one intended task, the MEC server only delivers the corresponding components which have not been stored in the VR device, and then the VR device uses the received components and the corresponding cached components to construct the task by exploiting the local computation resource.
  \item \emph{We develop an optimal task scheduling policy to minimize the average transmission data per task}. Of course, the MEC server can also select MEC computation model, which the MEC server combines all corresponding contents as the target task and then deliver the task to the VR device. The MEC computation mode is a reliable way to reduce the latency due to the fast computation capability at the MEC server, but this model delivers more data per task to the user. Therefore, we formulate the transmission data consumption minimization problem under the delay constraint and propose a task scheduling strategy based on Lyapunov theory.
  \item \emph{We discuss the tradeoffs between communications, computing, and caching}. We present how to joint allocate communications, computing and caching resources in the proposed communications-constrained MEC system to achieve a target end-to-end latency, e.g., $20~ms$.
\end{itemize}

\section{System Model}
As shown in Fig.\ref{System}, we consider an MEC system, where a cache-enabled mobile VR device can access BS with an MEC server to obtain task. The MEC server has an abundance of computing and caching resource, while the mobile VR device has limited computing ability and cache capacity.
\subsection{Task Model}

\begin{figure}
\centering
\includegraphics[width=3.3in, height=1.5in]{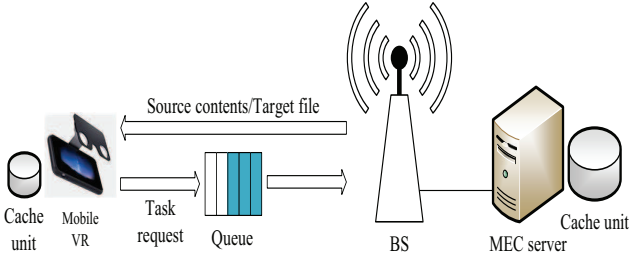}
\caption{The proposed communication-constrained MEC system with a caching enabled mobile VR device}\label{System}
\vspace{-4.5mm}
\end{figure}
We consider each task consists of a number of contents, e.g., MMT assets. All the contents composing each task come from a set of $N$ possible contents, which is denoted by $\mathbb{F}=\{F_1, F_2......F_N \}$. Note that one content may be used more than once in a task. The popularity distribution of the contents is denoted by $\mathbf{p}=[p_1,\ldots,p_N]$, where $\sum\limits_{i=1}^{N}p_i=1$. We assume all the contents are of equal size $\tau$ and the MEC server has all $N$ contents. The cache capacity of mobile VR device is $M$ with $M<N$, which can store at most $M$ contents. We adopt the most popular caching strategy and the stored content set can be denoted as $\mathbb{M}=\{F_1, F_2......F_M \}$.

The system is time-slotted with the time slot length $\Delta$. Let $H_t$ be the task scheduled at time slot $t$, which consists of $K_t$ contents. We denote $H(t)=[h_{1}(t),\ldots,h_{K_t}(t)]$ as the content index vector of the task $H_t$, where $h_{k_t}(t) \in\{1,\ldots,N\}$ indicates that the $k_t$-th content in $H_t$ is $F_{h_{k_t}(t)}$. Thus the size of $H_t$ is $D(t)=\tau{K_t}$. Let $G_n^t(1\le n \le N)$ denote whether $F_n\in \mathbb{F}$ is requested in $H_t$ and not cached in mobile VR device, which can be given by
\begin{equation}
G_n^t= \\
 \begin{cases}
 1,~~~~ \text{for}~ M+1\le n\le N, ~\text{and} ~h_{k_t}(t)=n\\
 0, ~~~~\text{otherwise}.
 \end{cases}
\end{equation}

\subsection{Computation Model}
Each task can either be executed at the MEC server, or at the mobile VR device. Let $W$ denote the required CPU cycles for computing one bit. The CPU frequency of the MEC server and the mobile VR device is $f_c$ and $f_l$, respectively. The wireless transmission throughput is $R$.
\subsubsection{MEC Computation Mode}
When the mobile VR device has a task request, the MEC server computes the contents as the task, and then deliver the task. The size of transmission data is $D_{ct} (t)=D(t)$. The size of computation data is $D_{cc}(t)=D(t)$. Therefore, $N_c (t)=\lceil D_{cc} (t)W/(f_c \Delta)+ D_{ct}(t)/(R \Delta)\rceil$ time slots are required to complete the task $H_t$. Similar to \cite{LiuJ}, we use $S_c (t)\in\{0, 1, ..., N_c (t)-1\}$ to model the MEC state, where $S_c (t)=0$ means that the MEC server is idle and $S_c (t)=n~(n\neq 0)$ indicates a task is processing at the MEC server and $N_c(t)-n$ time slots are required to complete the computation.
\subsubsection{Local Computation Mode}
For this model, the MEC server delivers the contents which do not been stored in the mobile VR device, and the mobile VR device computes the received contents with the cached contents as the task. Because the MEC server does not transmit the stored contents of the task, the size of transmission data is $D_{lt}(t)=\tau\sum\limits_{n=M+1}^{N}G_n^t$. The size of computation data also is $D_{lc}(t)=D(t)$. As a result, $N_l (t)=\lceil D_{lc}(t)W/(f_l \Delta)+ D_{lt}(t)/(R \Delta)\rceil$ time slots are required to complete $H_t$. Similarly, let $S_l (t)\in \{0, 1,...,N_l(t)-1\}$ denote the mobile VR device state. If the system allocates a task to the mobile VR device at time slot $t$, $S_l (t)$ updates to $S_l (t+1)=N_l (t)-1$.   
\subsection{Task Queueing Model}
The task arrival process is modeled as a Bernoulli process with probability $\lambda$. When a task arrivals, the task first enters into a task queue with infinite capacity. The number of the task waiting in the queue is the queue state $Q(t)=\{0, 1, 2, 3......\}$, where $Q(t)$ updates according to the following equation
\begin{equation}\label{queue}
  Q (t+1){\setlength\arraycolsep{0.3pt}=}(Q(t)-(u_l^1(t)+u_l^2(t)+u_c^1(t)+u_c^2(t)))^++A(t),
\end{equation}
where $A(t)$ denotes whether a task arriving in the time slot $t$. Thus we have $\Pr\{A(t)=1\}=\lambda$ and $\Pr\{A(t)=0\}=1-\lambda$. Here, $\{u_l^1(t), u_l^2(t),u_c^1(t),u_c^2(t)\}$ denotes the task scheduling decision at time slot $t$.

Notice that at most two tasks can be scheduled at a time slot. The first task should be scheduled before the second task. If the first task is scheduled to do the computation at the mobile VR device (MEC server), we have $u_l^1(t)=1~(u_c^1(t)=1)$ and the second task can not be scheduled to operate in the local computation mode (MEC computation mode) because the CPU has been occupied by the first task, yielding $u_l^2=0~(u_c^2(t)=0)$. Otherwise, we have $u_l^1(t)=0~(u_c^1(t)=0)$ for the first task and the local computation mode (MEC computation mode) could be scheduled for the second task, i.e., $u_l^2=1~(u_c^2(t)=1)$. As a result, there are five possible states for the task scheduling decision in each time slot, i.e., $\{u_l^1(t), u_l^2(t),u_c^1(t),u_c^2(t)\}=\{(0,0,0,0), (1,0,0,0), (0,0,1,0), (1,0,0,1),(0,1,1,0)\}$.

\section{Task Scheduling Strategy and Problem Formulation}
The MEC server is a reliable way to reduce the computation latency due to $f_c\geq f_l$, but consumes more communication resource due to $D_c (t)\geq D_l (t)$. Hence, the MEC server needs to make the task scheduling decision at each time slot to minimize the average transmission data per task under the average delay constraint.

\subsection{Task Scheduling Strategy}
When the MEC server (mobile VR device) is idle, the task can be scheduled to the MEC computation mode (local computation mode). The queue state $Q(t)=0$ denotes the task queue is empty and there is no task will be scheduled in time slot $t+1$, while $Q(t)=\infty$ indicates that there are infinite tasks in the task queue, yielding the unstable system. According to $Q(t)$, $S_l (t)$, and $S_c (t)$, we can describe the system state.

\textbf{Case 1}: $S_l (t)=S_c (t)=0$.
Both the mobile VR device and the MEC server are idle. The system can process at most two tasks. If there are two tasks in $Q(t)$ at least, i.e., $Q(t)\ge 2$, one task can be processed in the mobile VR device (MEC server) and the other task remains wait in the task queue or to be precessed in the MEC server (mobile VR device). The task scheduling policy can be expressed as the follow:
\begin{equation}\label{case1}
 u^1(t) = \\
 \begin{cases}
 (0,0,0,0)\\
 (0,0,1,0)\\
 (1,0,0,0)\\
 (1,0,0,1) \\
 (0,1,1,0)
 \end{cases} \text{for}~Q(t)\ge 2.
\end{equation}

If there is only one task in $Q(t)$, the task can be processed in the mobile VR device, the MEC server or remains wait in the task queue. We thus have
\begin{equation}\label{case2}
 u^2(t) = \\
 \begin{cases}
 (0,0,0,0)\\
 (0,0,1,0)\\
 (1,0,0,0)
 \end{cases}  \text{for}~Q(t)=1.
\end{equation}

\textbf{Case 2}:  $S_l (t)\neq 0, S_c (t)=0$.                
In this case, the MEC server is idle and the mobile VR device is busy so that the system can process one task at most for the MEC server. The task scheduling policy is:
\begin{equation}\label{case3}
 u^3(t) = \\
 \begin{cases}
 (0,0,0,0)\\
 (0,0,1,0)\\
 \end{cases}\text{for}~Q(t)\ge 1.
\end{equation}

\textbf{Case 3}: $S_l (t)=0, S_c (t)\neq 0$.
In this case, the mobile VR device operates in idle mode and the MEC server is occupied. Only one task can be scheduled for the mobile VR device. The task scheduling policy can be expressed as following:
\begin{equation}\label{case4}
 u^4(t) = \\
 \begin{cases}
 (0,0,0,0)\\
 (1,0,0,0)\\
 \end{cases}\text{for}~Q(t)\ge 1.
\end{equation}

\textbf{Case 4}: $S_c (t)\neq 0,  S_l (t)\neq 0$ or $Q(t)=0$.
If both the MEC server and mobile device are busy, i.e., ($S_c (t)\neq 0,  S_l (t)\neq 0$), or there is no task in the task queue $Q(t)=0$, no task is scheduled. We then have
\begin{equation}\label{case5}
  u^5(t) =(0,0,0,0).
\end{equation}

At the time slot $t$,  $S_l (t)$ and $S_c (t)$ can be expressed as:
\begin{equation}
 S_l(t+1) = \\
 \begin{cases}
 \max (S_l(t){\setlength\arraycolsep{0.3pt}-}1, 0) &u_l^1(t)=0 ~\text{or}~ u_l^2(t){\setlength\arraycolsep{0.3pt}=}0,\\
 N_l^1(t)-1   &u_l^1(t)=1, \\
 N_l^2(t)-1   &u_l^2(t)=1. \label{local state}
 \end{cases}
\end{equation}

\begin{equation}
 S_c(t{\setlength\arraycolsep{0.3pt}+}1) {\setlength\arraycolsep{0.3pt}=} \\
 \begin{cases}
 \max (S_c(t){\setlength\arraycolsep{0.3pt}-}1, 0) &u_c^1(t){\setlength\arraycolsep{0.3pt}=}0 ~\text{or} ~u_c^2(t){\setlength\arraycolsep{0.3pt}=}0,\\
 N_c^1(t)-1   &u_c^1(t)=1, \label{mec state}\\
 N_c^2(t)-1   &u_c^2(t)=1.
 \end{cases}
\end{equation}
where $N_l^i(t)$ and $N_c^i (t)$ denote $N_l(t)$ and $N_c(t)$ of the $i$-th task for $i=1,2$, respectively. $S_l (t+1)=N_l^i (t)-1$ means the mobile VR device is occupied by a task in the time slot $t+1$ and will be busy in the follow $N_l^i (t)-1$ time slots. Similarly, we have $S_c (t+1)=N_c^i(t)-1$. 
\subsection{Problem Formulation}
When $T\to \infty$ and the length of task queue is not infinite, the total number of the task is close to $\lambda T$. Therefore, the average transmission data per task can be expressed as:
\begin{equation}\label{average data}
 \lim_{T\to\infty}  \frac{1}{\lambda T} \Big\{\sum_{t=0}^{T-1}\sum_{i=1}^{2}u_l^i(t) D_{lt}^i (t)+u_c^i (t) D_{ct}^i (t)\Big\}.
\end{equation}
where $D_{lt}^i(t)$ and $D_{ct}^i (t)$ denote $D_{lt}(t)$ and $D_{ct}(t)$ of the $i$-th task for $i=1,2$, respectively.

From the system model, we know that each task requires transmission time, waiting time and processing time. The computation processing time of the MEC server or the mobile VR device is the dominant influence on the execution delay \cite{Lyapunov}. Based on the Little Law \cite{queue}, the execution delay, including the waiting time and processing time, is proportional to the average queue length of the task buffer. The execution delay is written as:

\begin{equation}\label{average delay}
 \lim_{T\to\infty}  \frac{1}{T} \mathbb{E}\Big[\sum_{t=0}^{T-1}Q(t)\Big].
 \end{equation}

Denote $\pi(t)\triangleq \{u_l^1(t), u_l^2(t),u_c^1(t),u_c^2(t)\}$. Thus, the communication-resource consumption minimization problem is formulated as:
\begin{align}
  &\mathrm{P1}:\min_{\pi(t)}~ \hspace{4mm}  \lim_{T\to\infty}  \frac{1}{\lambda T} \Big\{\sum_{t=1}^{T}\sum_{i=1}^{2}u_l^i(t) D_{lt}^i (t)+u_c^i (t) D_{ct}^i (t)\Big\}  \nonumber\\
  & \hspace{8mm}s.t.                                                                                                                                                                                                                                                                                                                                                                                                                                                                                                                                                                                                                                                                                                                                                                                                                                                                                                                                                                                                                                                                                                                                                                                                                                                                                                                                                                                                                                                                                                                                                                                                                                                                                                                                                                                                                                                                                                                                                                                                                                                                                                                                                                                                                                                                                                                                                                                                                                                                                                                                                                                                                                                                                                                                                                                                                                                                                                                                                                                                                                                                                                                                                                                                                                                                                                                                                                                                                                                                                                                                                                                                                                                                                                                                                                                                                                                                                                                                                                                                                                                                                                                                                                                                                                                                                                                                                                                                                                                                                                                                                                                                                                                                                                                                                                                                                                                                                                                                                                                                                                                                                                                                                                                                                                                                                                                                                                                                                                                                                                                                                                                                                                                                                                                                                                                                                                                                                                                                                                                                                                                                                                                                                                                                                                                                                                                                                                                                                                                                                                                                                                                \hspace{11mm}\pi(t)\in u^k(t), ~~~~k\in \{1, 2, 3, 4, 5\}, \label{u_l^k} \\
  &\hspace{22mm}\lim_{T\to\infty}  \frac{1}{T} \mathbb{E}\Big[\sum_{t=1}^{T}Q(t)\Big]< \infty,\qquad\quad  \label{delayconstraint}
\end{align}
where (\ref{delayconstraint}) indicates the delay constraint to ensure the task requires can be completed with a finite delay. Unfortunately, $\mathrm{P1}$ is a stochastic optimization problem. The system state changes after a offloading decision is made, and $\mathrm{P1}$ is impossible to be solved by convex optimization methods.
\section{Optimal Task Scheduling Algorithm Based On Lyapunov Theory}
\subsection{Problem Transform}
To simplify $\mathrm{P1}$,  we consider Lyapunov optimization theory. We first define the Lyapunov function:
\begin{equation}\label{Lyapunove function}
 L(Q(t))=\frac{1}{2} Q^2(t).
\end{equation}

Consider the initial state $Q(0)=0$, and then we have $L(Q(0))=0$. If the queue is unstable, $L(Q(t))$ is more volatile than $Q(t)$. Thus the expectation of $L(Q(t))$ is:
 \begin{align}\label{Lyapunov expectation}
  &\mathbb{E}[L(Q(t))] = \mathbb{E}\Big\{\sum_{i=0}^{t-1}[L(Q(i+1))-L(Q(i))]\Big\} \nonumber\\
  &\hspace{16mm}=\sum_{i=0}^{t-1}\mathbb{E}\{L(Q(i+1))-L(Q(i))|Q(i)\}.
\end{align}
The system is stable when $\mathbb{E}[L(Q(t))]< \infty$. Therefore the Lyapunov drift function can be given by:
 \begin{equation}\label{Lyapunov drift}
 \Delta L(Q(t))=\mathbb{E}\Big\{L(Q(t+1))-L(Q(t))|Q(t)\Big\}.
\end{equation}

We can see from (\ref{Lyapunov expectation}) and (\ref{Lyapunov drift}) that to maintain the stability of the queue, we should minimize (\ref{Lyapunov drift}) in each time slot. Therefore the expectation of the $L(Q(t))$ would not tend to infinite. As a result, we have the following Lemma 1.
\begin{lemma}
Let us define the scheduling rate $U(t)=u_l^1(t)+u_l^2(t)+u_c^1(t)+u_c^2(t)$. In order to ensure $\mathbb{E}[L(Q(t))]< \infty$, we have:
\begin{equation}
 \Delta L(Q(t))\le C_{max}-Q(t)\mathbb{E}[U(t)|Q(t)],\label{Lyapunov drift2}
\end{equation}

where we use $C_{max}=\frac{1}{2} \{5+2Q(t)A(t) \}$.

\end{lemma}

\begin{proof}
Please refer to Appendix \ref{1}.
\end{proof}

According to Lyapunov theory \cite{queue}, when we make the task scheduling decision $\pi(t)$ to minimize $\Delta L(Q(t)$, the queue state $Q(t)$ can also approach a lower length. However, the minimization of $\Delta L(Q(t)$ can not cause the minimization of (\ref{average data}). Thus, we define the Lyapunov drift-plus-penalty function:
\begin{align}
\Delta L(Q(t))+V\mathbb{E}[D(t)|Q(t)]\le &C_{max}-Q(t)\mathbb{E}[U(t)|Q(t)] \nonumber\\
&+V\mathbb{E}[D(t)|Q(t)],\label{drift-plus-penaltyF}
\end{align}
where $V$ is a non-negative control parameter, which denotes that the system is sensitive to the communication cost. When $V=0$, the system is only sensitive to the delay. With the increase of $V$, the Lyapunov drift-plus-penalty becomes more sensitive to the communication cost. Notice that the optimal task scheduling decision $\pi^*(t)$ for minimizing the right side of (\ref{drift-plus-penaltyF}) also minimize $D(t)$ under the queue length stability constraint. Therefore, we can solve $\mathrm{P2}$ in each time slot $t$:
\begin{equation}
 \mathrm{P2}:\min~ \hspace{4mm} -Q(t)U(t)+VD(t)
\end{equation}
\begin{align*}\label{P2}
  & \hspace{8mm}s.t.                                                                                                                                                                                                                                                                                                                                                                                                                                                                                                                                                                                                                                                                                                                                                                                                                                                                                                                                                                                                                                                                                                                                                                                                                                                                                                                                                                                                                                                                                                                                                                                                                                                                                                                                                                                                                                                                                                                                                                                                                                                                                                                                                                                                                                                                                                                                                                                                                                                                                                                                                                                                                                                                                                                                                                                                                                                                                                                                                                                                                                                                                                                                                                                                                                                                                                                                                                                                                                                                                                                                                                                                                                                                                                                                                                                                                                                                                                                                                                                                                                                                                                                                                                                                                                                                                                                                                                                                                                                                                                                                                                                                                                                                                                                                                                                                                                                                                                                                                                                                                                                                                                                                                                                                                                                                                                                                                                                                                                                                                                                                                                                                                                                                                                                                                                                                                                                                                                                                                                                                                                                                                                                                                                                                                                                                                                                                                                                                                                                                                                                                                                                \hspace{10mm}(\ref{u_l^k})~ \text{and} ~(\ref{delayconstraint}). \qquad\quad
\end{align*}

For each time slot $t$, we can obtain $D(t)$ based on $\pi(t)$ and $Q(t)$. Because there are only 5 possible choices for $\pi(t)$, we can solve $\mathrm{P2}$ in each time slot $t$ by an enumeration method. Thus, we propose an optimal task scheduling algorithm based on Lyapunov theory, as shown in Algorithm 1.

\begin{algorithm}[t!]
\caption{Optimal Task Scheduling Algorithm Based On Lyapunov Theory}
\label{Algorithm2}
\begin{algorithmic}[1]
\STATE \textbf{Obtain the queue state $Q(t)$,  mobile device state $S_l(t)$, MEC server state $S_c(t)$ at the beginning of each time slot $t$}.
\STATE \textbf{Find the system case discussed in Section III}.
\STATE \textbf{Obtain the system case $k$}.
\STATE \textbf{Determine $\pi(t)$ by solving}:
\STATE \hspace{5mm}  $min$   \hspace{5mm}  $-Q(t)U(t)+VD(t)$
\STATE \hspace{5mm} $s.t.$   \hspace{6mm} \textbf{$(\ref{u_l^k}), ~and ~(\ref{delayconstraint})$}
\STATE \textbf{Set $t=t+1$ and update $Q(t)$,  $S_l(t)$ ,  $S_c(t)$ according to (\ref{queue}), (\ref{local state}), (\ref{mec state}) respectively}.
\end{algorithmic}
\end{algorithm}

\begin{lemma}\label{optF}
$\mathrm{P1}$ is not equivalent to $\mathrm{P2}$, but if the control parameter $V$ is sufficiently large, the solution of $P1$ is very close to $\mathrm{P2}$. Let $\bar D^{Alg}$ and $\bar D^{Opt}$ be the average transmission data obtained by solving $\mathrm{P2}$ and the optimal value of $P1$, respectively. We then have:
\begin{equation}\label{Cmaxuse}
\bar D^{Opt} \le \bar D^{Alg} \le \frac{5}{2V}+\bar D^{Opt}.
\end{equation}
\end{lemma}
\begin{proof}
Please refer to Appendix \ref{2}.
\end{proof}
\subsection{Tradeoffs between Communications, Computing and Caching}
In this subsection, we reveal the tradeoff between the average transmission data per task and the computing ability of the mobile VR device. The tradeoff between the  the average transmission data per task and the cache capacity is also discussed.

\begin{proposition}\label{aaF}
Let $\mathbb{E}[D_{ct}(t)]=\bar D_{ct}$ and $\mathbb{E}[D_{lt}(t)]=\bar D_{lt}$ denote the average transmission data of MEC computation model and local computation model, respectively. And let $\mathbb{E}[D_{lc}(t)]=\bar D_{lc}$ denote the average computation data of local computation model. We thus have
\begin{equation}
 \bar D^{Opt} {\setlength\arraycolsep{0.3pt}=} \\
 \begin{cases}
 \bar D_{ct} -\frac{1}{\lambda \bar N_l}(\bar D_{ct}-\bar D_{lt})  &\frac{1}{\bar N_l}+\frac{1}{\bar N_c} \ge \lambda,\frac{1}{\bar N_l} < \lambda,\\
 \bar D_{lt}   &\frac{1}{\bar N_l} \ge \lambda, \\
 infeasible   &\frac{1}{\bar N_l}+\frac{1}{\bar N_c} < \lambda.\label{Data_localcomputeF}
 \end{cases}
\end{equation}
 where we use $\mathbb{E}[N_c(t)]=\bar N_c$ and $\mathbb{E}[N_l(t)]=\bar N_l$ to denote the exception of the time slots required to complete a task in MEC computation model and local computation model, respectively.
\end{proposition}
\begin{proof}
Due to page limitations, we skip the proof here.
\end{proof}

According to Lemma \ref{optF}, with large $V$, $\bar D^{Alg}$ approaches $\bar D^{opt}$, and $\bar D^{Alg} \le \frac{5}{2V}+\bar D^{Opt}$ based on (\ref{Cmaxuse}). Because $\bar N_l=\mathbb{E}[N_l(t)]=\mathbb{E}[\lceil D_{lc}(t)W/(f_l \Delta)+ D_{lt}(t)/R\rceil]$, for large $V$, we can observe from (\ref{Data_localcomputeF}) that $\bar D^{Alg}$ or $\bar D^{opt}$ decreases with the increase of $f_{l}$ when $\frac{1}{\bar N_l}+\frac{1}{\bar N_c} \ge \lambda,\frac{1}{\bar N_l} < \lambda$.

\begin{proposition}\label{cache_size}
Define $K=\mathbb{E}[K_t]$, and then we have
 \begin{align}
 &\mathbb{E}[D_{ct}(t)]=\bar D_{ct}=\tau K,  \label{Dct} \\
 &\mathbb{E}[D_{lt}(t)]=\tau \sum\limits_{n=M+1}^{N}\sum\limits_{k} 1-(1-p_n)^k\Pr(K_t=k). \label{Dlt}
 \end{align}
\end{proposition}
\begin{proof}
Due to page limitations, we skip the proof here.
\end{proof}
%
%
%

From proposition 2, we can see that $ \bar D_{lt}$ decreases with the increase of the caching capacity $M$ in the mobile VR device. And we can know that $\bar D^{Opt}$ decreases with the decrease of $D_{lt}$ based on Proposition 1. Hence the increase of $M$ can decrease $D^{Opt}$. Similar to proposition 1, with large $V$, $\bar D^{Alg}$ approaches $\bar D^{opt}$, and $\bar D^{Alg}$ decrease with the increase of $M$.

\section{Numerical Results}
In this section,  we evaluate the performance of the proposed optimal scheduling policy by simulations. The number of contents is $N=1000$, the time slot is $\Delta=0.2~s$, the length of each content is $\tau=5~Mbits$, and the average arrival rate is $\lambda=0.4$. The mobile VR device has $f_l=1~GHz$ CPU frequency, the cache capacity $M=50$ and transmission throughput $R=500Mbps$, when the MEC server has $f_c=4\times2.5~GHz$ CPU frequency unless otherwise specified. $K_{t}$ is distributed uniformly in $[40,$ $60]$. We assume the content popularity distribution is identical among all elements of a task, which follows the Zipf distribution with the parameter $\alpha$. We set $\alpha =0.8$ in simulations. We consider the \emph{MEC computation policy} as baselines, which executes all the tasks in the MEC server.

\begin{figure}
\centering
\includegraphics[width=3.1in,height=2.2in]{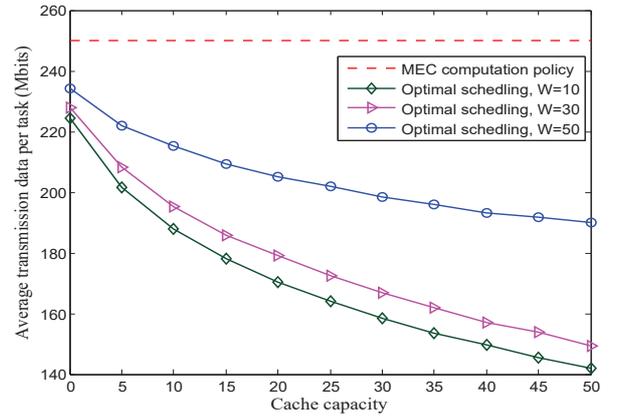}
\vspace{-2mm}
\caption{The communication-resource consumption vs. cache capacity.}\label{SRvsNtwithdiffLI} \label{cachep}
\label{cache2}
\vspace{-2mm}
\end{figure}

Fig.\ref{cache2} shows that the average transmission data per task achieved by the proposed optimal scheduling policy \emph{decreases} with the cache capacity. That means the average communication-resource consumption can be traded off by the cache capacity to keep the queue length stable, which verifies the tradeoff presented by Proposition 2. Moreover, the scheduling policy always outperforms the MEC computation policy even when there are no contents cached in the mobile VR device. This is because the optimal scheduling always executes a part of tasks by local computation policy, and the redundant transmission of contents needed in those tasks can be avoided. However, the MEC computation policy delivers all the contents requested by each task.

\begin{figure}
\centering
\includegraphics[width=3.1in,height=2.2in]{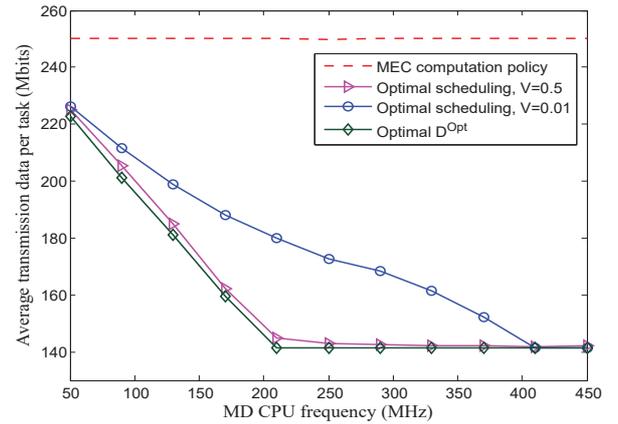}
\vspace{-2mm}
\caption{The communication-resource consumption vs. the computing ability of mobile VR device. } \label{Cost}
\vspace{-4.5mm}
\end{figure}

Fig.\ref{Cost} presents the average transmission data per task versus the mobile VR device computing ability $f_l$. The increase of the mobile VR device computing ability $f_l$ \emph{decrease} the average communication-resource consumption. The reason for it is that the increase of $f_l$ decreases the computing delay of local execution, and more tasks will be executed by local computation policy. With large $V$, the system is more sensitive to the communication-resource consumption, and more tasks are scheduled to mobile VR device. And $D^{Alg}$ is close to $D^{Opt}$ when $V$ is sufficiently large, which verify the lemma 2. Further, when $f_l$ is sufficiently large, the local computation mode becomes the optimal schedule strategy.
\begin{figure}
\centering
\includegraphics[width=3.1in,height=2.2in]{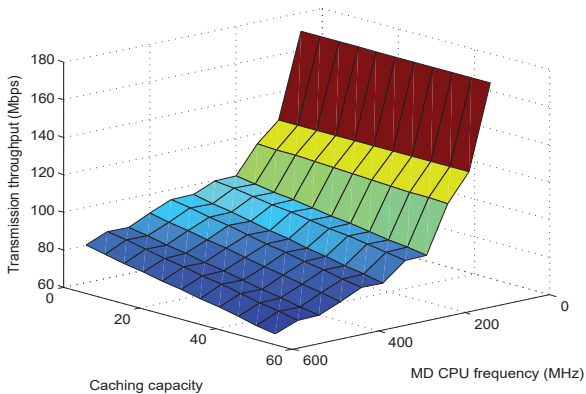}
\vspace{-2mm}
\caption{The incorporation of communications, computing and caching to achieve a target delay $20~ms$. Here, we set $\lambda =0.2$. }\label{SRvsNtwithdiffLI}
\label{3C}
\vspace{1.5mm}
\end{figure}

 \textbf{Tradeoff between Communications, Computing and Caching (3C)}: In Fig. \ref{3C}, the average end-to-end latency is $20~ ms$ of the proposed system with different 3C resources allocation, e.g., $\{R,f_l,M\}$=$\{125~Mbps, 100~MHz, 5\}$ or $\{R,f_l,M\}$=$\{67~Mbps, 500~MHz, 45\}$. As one can see, the communication throughput $R$ decreases with increasing computing capability $f_l$ and caching capacity $M$. As $f_l$ increases, more task be scheduled to the mobile VR device, yielding lower the communication cost. $M$ is similar to $f_l$. Therefore, the system needs a smaller $R$. We observe that the computing ability has more impact on the communication-resource consumption than that of the caching capacity. This is because the cache capacity only has impact on the performance of local computation mode. On the other hand, when the system has small computing ability and caching capability, the large transmission throughput is required. This result reconfirm the importance of simultaneous exploitation of all 3Cs in wireless networks \cite{3C}.

\

\section{Conclusion}
In this paper, we investigated the communication-constrained mobile edge computing systems for wireless virtual reality. A transmission data consumption minimization problem with the execution delay constraints was formulated, and we proposed a task scheduling strategy based on Lapunov theory. The tradeoffs between communications, computing, and caching in the proposed system was also dicussed. Finally Simulation results shown that the proposed scheduling strategy achieve a significant reduction in the average transmission data consumption.  
\appendix
\subsection{Proof of Lemma 1}\label{1}
According to (\ref{queue}), we first have
 \begin{equation}\label{Q(t+1)}
Q^2(t+1)\le Q^2(t)+U^2(t)+A^2(t)-2Q(t)(U(t)-A(t)).
\end{equation}
Substituting (\ref{Q(t+1)}) into (\ref{Lyapunov drift}), then (\ref{Lyapunov drift}) can be rewritten as
\begin{align}
\Delta L(Q(t)) &\le \frac{1}{2} \mathbb{E}[U^2(t)+A^2(t)|Q(t)]\nonumber \\
   &-Q(t)\mathbb{E}[(U(t)-A(t))|Q(t)]\label{Lyapunov drift3}.
\end{align}
The task arriving rate $A(t)$ is independent of $Q(t)$. Finally, according to $U(t)\leq2$ and $A(t)\leq1$, we have Lemma 1.
\subsection{Proof of Lemma 2}\label{2}
We assume $P1$ is feasible, and there exists at least one $\pi^*(t)$ for satisfying the constraints of $P1$. $\bar D^{Alg}$ and $\bar D^{Opt}$ satisfies the following condition:
 \begin{equation}\label{DaleDo}
\mathbb{E}[D(t)|Q(t)]=\bar D^{Alg} \leq \bar D^{Opt}+ \gamma,
\end{equation}
where $\gamma$ is a positive value. According to Little Theorem \cite{queue},  if the average arriving rate is larger than the average service rate, the queue length tends to infinite with the increase of time slot $t$. Therefore, if $P2$ can be solved by proposed algorithm, the following condition should be satisfied
\begin{equation}\label{segearF}
\mathbb{E}[U(t)|Q(t)]= \lambda+\theta,
\end{equation}
where $\theta$ is a positive value. Substituting (\ref{DaleDo}) and (\ref{segearF}) into (\ref{drift-plus-penaltyF}), and with $\gamma \to 0$, we obtain:
\begin{equation}\label{drift-plus-penalty2}
 \Delta L(Q(t))+V\{D^{Alg}(t)|Q(t)\}\le \frac{5}{2}+V\bar D^{opt}.
\end{equation}

Then taking iterated expectation and using the telescoping sums over $t \in \{1......T\}$,  we get
\begin{align}\label{drift-plus-penalty2}
 &\mathbb{E}[L(Q(T))]-\mathbb{E}[L(Q(1))]+V\sum_{t=1}^{T}\mathbb{E}[D^{Alg}(t)|Q(t)]\nonumber\\
  &\leq T(\frac{5}{2}+V\bar D^{opt}).
\end{align}

We divide (\ref{drift-plus-penalty2}) with $VT$ and let $T\to \infty$, then we have:

\begin{equation}\label{Cmax}
\bar D^{Alg} \le \frac{5}{2V}+\bar D^{Opt}.
\end{equation}

\bibliographystyle{IEEEtran}
\bibliography{paper}

\end{document}